\documentclass[ 
aip,
jmp,
longbibliography,
draft, % bmf,
 amsmath,amssymb,
reprint, onecolumn %%%%% Comment onecolumn to override one column format %%%%%%%
]{revtex4-1}
\usepackage{amsthm,amsmath,amssymb}
\usepackage{mathtools}
\usepackage{cleveref}
\newtheorem{definition}{Definition}[section]

\newtheorem{proposition}[definition]{Proposition}
\newtheorem{lemma}[definition]{Lemma}
\newtheorem{theorem}[definition]{Theorem}
\newtheorem{corollary}[definition]{Corollary}
\newtheorem{remark}[definition]{Remark}
\newtheorem{example}[definition]{Example}
\begin{document}

% Use the \preprint command to place your local institutional report number 
% on the title page in preprint mode.
% Multiple \preprint commands are allowed.
%\preprint{}

\title{The Clifford theory of the n-qubit Clifford group} %Title of paper

% repeat the \author .. \affiliation  etc. as needed
% \email, \thanks, \homepage, \altaffiliation all apply to the current author.
% Explanatory text should go in the []'s, 
% actual e-mail address or url should go in the {}'s for \email and \homepage.
% Please use the appropriate macro for the type of information

% \affiliation command applies to all authors since the last \affiliation command. 
% The \affiliation command should follow the other information.

\author{Kieran Mastel}
\affiliation{ 
Institute for Quantum Computing, University of Waterloo, Waterloo, Ontario N2L 3G1, Canada}%
\affiliation{Department of Pure Mathematics, University of Waterloo, Waterloo, Ontario N2L 3G1, Canada}
\altaffiliation{Now at the Department of Mathematics and Statistics, University of Ottawa, Ottawa, Canada}
\email{kmastel@uottawa.ca}
%\email[]{Your e-mail address}
%\homepage[]{Your web page}
%\thanks{}
%\altaffiliation{}
%\affiliation{}

% Collaboration name, if desired (requires use of superscriptaddress option in \documentclass). 
% \noaffiliation is required (may also be used with the \author command).
%\collaboration{}
%\noaffiliation

\date{\today}

\begin{abstract}
The n-qubit Pauli group and its normalizer the n-qubit Clifford group have applications in quantum error correction and device characterization. Recent applications have made use of the representation theory of the Clifford group. We apply the tools of (the coincidentally named) Clifford theory to examine the representation theory of the Clifford group using the much simpler representation theory of the Pauli group. We find an unexpected correspondence between irreducible characters of the n-qubit Clifford group and those of the (n+1)-qubit Clifford group.
\end{abstract}

% insert suggested PACS numbers in braces on next line

\maketitle %\maketitle must follow title, authors, abstract and \pacs

% Body of paper goes here. Use proper sectioning commands. 
% References should be done using the \cite, \ref, and \label commands
\section{\label{sec:intro}Introduction}

The Pauli group and its normalizer, the Clifford group, are fundamental structures in quantum information theory. These groups have applications in quantum error correction \cite{Gottesman2010} and randomized benchmarking \cite{Helsen_2019}. By the Gottesman-Knill theorem, quantum computation using Clifford unitaries is efficiently simulable on a classical computer \cite{Gottesman1998,PhysRevA.70.052328}. The Clifford group is a unitary 2-design \cite{Dankert_2009}, in other words, `averages over the Clifford group approximate averages over the unitary group well.' Generating random Clifford unitaries is less computationally expensive than sampling Haar random unitaries \cite{Koenig_2014}. Thus, random Clifford elements have utility in performing randomized protocols. Recent applications of the Clifford group to randomized benchmarking and classical shadow estimation have utilized its representation theory \cite{Helsen_2019,helsen2022thrifty,Helsen2018}. Determining the character table of the Clifford group, which classifies its irreducible representations, is a natural open problem prompted by these papers. Surprisingly, despite the usefulness of the representation theory of the Clifford group, its character table has not been determined.

The representation theory of the Pauli group is simple and explained in \cref{sec:paulirep}. Thus, it would be advantageous to use our understanding of the representation theory of the Pauli group to examine that of the Clifford group. To do this, we can apply the tools of Clifford theory (which is named after Alfred H. Clifford, while William K. Clifford gives his name to the group). Clifford theory is the part of representation theory focused on relating representations of a normal subgroup $N$ of $G$ to representations of $G$. The inertia subgroup $I_{G}(\sigma)$ is the subgroup of $G$ that maps a representation $\sigma$ of $N$ to an isomorphic representation under conjugation. The central result of Clifford theory is the \textit{Clifford correspondence} between irreducible representations of the inertia subgroup and certain irreducible representations of $G$. When the inertia subgroup is understood, this simplifies the calculation of irreducible characters of $G$. Since any two nontrivial irreducible Pauli representations are conjugate in the Clifford group and conjugate representations have isomorphic inertia subgroups, we need only examine one inertia subgroup. In our first result, we determine the inertia subgroup of a nontrivial irreducible representation of the 
$n$-qubit Pauli group in the $n$-qubit Clifford group up to complex phases for $n\geq 2$.

Clifford theory does not fully calculate the character table of the Clifford group. The Clifford correspondence does not give us any information when the inertia subgroup $I_G(\sigma)$ is all of $G$. In particular, the Clifford correspondence does not help when $\sigma$ is the trivial representation of $N$. For a group $G$, inflation produces a bijection between irreducible representations of the quotient group $G/N$ and irreducible representations of $G$ whose restriction to $N$ is trivial. We can thus understand the case where the Clifford correspondence offers no information by examining the representation theory of the quotient group. For the $n$-qubit Clifford group and the $n$-qubit Pauli group, the quotient group is the symplectic group $Sp(2n,2)$. The symplectic group is a finite group of Lie type, and thus its representation theory is calculated by the Deligne-Lusztig theory \cite{geck_malle_2020}, which we do not examine in this paper. Together with the representations calculated using Clifford theory, this accounts for all the irreducible representations of the Clifford group.

In \cref{sec:quotrep}, we show that the inertia quotient group $I_G(\sigma)/N$ of a nontrivial Pauli representation in the $n$-qubit Clifford group is a central extension of the affine symplectic group $Sp(2(n-1),2)\ltimes\mathbb{Z}^{2(n-1)}_2$ by $\mathbb{Z}_2$. The Clifford group, in the literature on finite group extensions, is known as the unique non-split extension of $Sp(2n,2)$ by $\mathbb{Z}^{2n}_2$. By examining the Clifford group from this perspective, Bernd Fischer showed in \cite{Fischer88} that the Clifford and affine symplectic groups have identical character tables. Combining these facts allows us to produce a surprising correspondence between irreducible characters of the $n$-qubit Clifford group $C_{n}$  and the $(n+1)$-qubit Clifford group $\mathcal{C}_{n+1}$. Any irreducible character of $\mathcal{C}_n$ can be viewed as an irreducible character of $Sp(2n,2)\ltimes\mathbb{Z}_2^{2n}$ which inflates to an irreducible character of the inertia subgroup $IN_{n+1}$ of a nontrivial representation of the $(n+1)$-qubit Pauli group in the $(n+1)$-qubit Clifford group. This representation induces an irreducible character of $\mathcal{C}_{n+1}$ by the Clifford correspondence. The natural map of characters from $\mathcal{C}_n$ to $\mathcal{C}_{n+1}$ is induction. However, unlike the induction of characters, our correspondence maps irreducible characters to irreducible characters. Knowing the irreducible characters of $\mathcal{C}_n$ allows us to calculate an equal number of the irreducible characters of $\mathcal{C}_{n+1}$.
%\label{}
\section{Preliminaries}
\subsection{Representation theory}
In this section, we recall from \cite{serre_1977} some basic facts about representation and character theory of finite groups. A  \textbf{linear representation} of a finite group $G$ is a homomorphism $\rho$ from the group $G$ into the group $GL(V)$, where $V$ is a vector space over $\mathbb{C}$.
If $W$ is a vector subspace of $V$ such that $\rho(g)x \in W$ for all $g \in G$ and $x \in W$, then the restriction $\rho^W(g)$ of $\rho(g)$ to $W$ is a linear representation of $G$ on $W$. $W$ is called a  \textbf{subrepresentation} of $V$.
An  \textbf{irreducible} representation is one where $V$ is not 0 and has no nontrivial subrepresentations. It is a standard result that every representation is a direct sum of irreducible representations.
 
If $\rho$ and $\sigma$ are representations of a finite group $G$ on the vector spaces $V$ and $W$ respectively then a linear map $\phi:V \rightarrow W$ is called an  \textbf{intertwining map} of representations if $\phi(\rho(g)v) = \sigma(g)\phi(v)$
for all $g \in G$ and all $v \in V$. The vector space of all such $G$-linear maps between $V$ and $W$ is denoted by $\text{Hom}_G(\rho,\sigma)$ or $\text{Hom}_G(V,W)$. If $\phi$ is also invertible, it is said to be an  \textbf{isomorphism} of representations. When we classify irreducible representations, we do so up to isomorphism. Isomorphic representations are sometimes called  \textbf{equivalent} representations.
 
Let $F$ be a field, a  \textbf{projective representation} of a finite group $G$ is a is a map $\Phi:G \rightarrow GL_n(F)$ such that for every $g,h \in G$, there exists a scalar $\alpha(g,h)\in F$ such that 
\begin{equation*}
    \Phi(g)\Phi(h) = \Phi(gh)\alpha(g,h).
\end{equation*}
The set of values $\alpha(g,h)$ is called the  \textbf{factor set}, and is uniquely determined by $\Phi$. The notions of equivalence and irreducibility translate verbatim for projective representations. We refer the reader to section 7.2 of \cite{ceccherini2022representation} for a more exhaustive discussion of projective representations.
 
Let $\rho: G \rightarrow GL(V)$ be a linear representation of a finite group $G$. For each $g \in G$, define
\begin{equation*}
    \chi_{\rho}(g) = \text{Tr}(\rho(g)),
\end{equation*}
with $\text{Tr}(\rho(g))$ being the trace of the operator $\rho(g)\in GL(V)$. The function $\chi_{\rho}$ on $G$ is called the  \textbf{character} of the representation $\rho$. If $\rho$ is irreducible, we call $\chi_{\rho}$ an irreducible character. It is a standard result that two representations are isomorphic if and only if they have the same character. Note that, from properties of the trace, $\chi_{\rho}(h^{-1}gh) = \chi_{\rho}(g)$ and thus characters are constant on the conjugacy classes of groups. In other terms, characters are  \textbf{class functions}. 
 
Since characters form an orthonormal basis of the space of class functions, the number of inequivalent irreducible representations equals the number of conjugacy classes of $G$. If $\chi$ is the character of a representation $(\rho,V)$ of $G$, and $e\in G$ is the identity, then $\chi(e) = \text{dim} V$ and is called the  \textbf{degree} of the character. If $G$ is abelian, then every character is of degree $1$.
 
The  \textbf{character table} of a finite group $G$ is the table with rows corresponding to inequivalent irreducible characters of $G$ and columns corresponding to conjugacy classes of $G$. Entry $(i,j)$ of the table is the value of the $i^{\text{th}}$ irreducible character of $G$ on the $j^{\text{th}}$ conjugacy class of $G$. 
 
Clifford theory deals with induced and restricted representations, which we will now define.
\begin{definition}
If $\rho$ is a representation of $G$ and $H$ is a subgroup of $G$ then we can define the  \textbf{restriction} of $\rho$ to $H$ 
\begin{equation*}
    (\mathrm{Res}_H^G\rho)(h) := \rho(h), \; \text{for all } \; h\in H.
\end{equation*}
The restriction is a representation of $H$ by definition. If $\chi$ is the character of $\rho$, we can also define the restriction of $\chi$ to $H$ by
\begin{equation*}
    (\mathrm{Res}^G_H\chi)(h) := \chi(h), \; \text{for all } \; h\in H.
\end{equation*}
Notice that $\mathrm{Res}^G_H\chi$ is the character of $\mathrm{Res}_H^G\rho$.
\end{definition}
Let $\rho$ be a representation of $G$ and $\psi$ be a subrepresentation of the restriction $\mathrm{Res}^G_H\rho$ of $\rho$ to a subgroup $H$ of $G$. Let $V$ and $W$ be the respective representation spaces of $\rho$ and $\psi$. For $s \in G$ the vector space $\rho(s)W$ depends only on the left coset $sH$ of $s$. Thus if $\gamma$ is a left coset of $H$ we can define the subspace $W_{\gamma}$ of $V$ to be $\rho(s)W$ for any $s \in \gamma$. Clearly, the $W_{\gamma}$ are permuted by $\rho(s)$ for $s \in G$. This tells us that $\sum_{\gamma\in G/H}W_{\gamma}$ is a subrepresentation of $V$.
\begin{definition}
We say that the representation $\rho$ of $G$ is \textbf{induced} by the representation $\psi$ of $H$ on $W$ if $V$ is equal to the direct sum of the $W_{\sigma}$ for $\sigma \in G/H$.
\end{definition}
Restriction and induction of representations do not preserve irreducibility in general. We state the following theorem from \cite{serre_1977} without proof.
\begin{theorem}
Let $(W,\psi)$ be a representation of $H$, and $H$ be a subgroup of $G$. There exists a linear representation of $G$ induced by $\psi$ which we denote $\mathrm{Ind}_H^G\psi$ or $\mathrm{Ind}_H^GW$. This induced representation is unique up to isomorphism.
\end{theorem}
\subsection{Clifford theory}
The objective of Clifford theory is to study the representation theory of a group via the representation theory of its normal subgroups. Here we review the central results of Clifford theory. In this section, we largely follow the outline of Clifford theory given in part 2 of \cite{Ceccherini-Silberstein2009}, and we refer the reader there for proofs and a more thorough exposition. In addition, we collect some results from \cite{ceccherini2022representation}, which will prove essential to our analysis.
 
Let $G$ be a finite group and $N \unlhd G$ be a normal subgroup of $G$. Let $\widehat{G}$ and $\widehat{N}$ denote the set of all irreducible representations of $G$ and $N$, respectively, up to equivalence. For two representations $\rho$ and $\sigma$ we write $\sigma \succeq \rho$ to denote that $\rho$ is a subrepresentation of $\sigma$ and $\rho \sim \sigma$ to denote that $\rho$ and $\sigma$ are isomorphic representations.
\begin{definition}
Let $\sigma \in \widehat{N}$ and $g \in G$. We define
\begin{equation*}
    \widehat{G}(\sigma) = \{\theta \in \widehat{G} :   Res^G_N(\theta)\succeq\sigma\}.
\end{equation*}
The \textbf{g-conjugate} of $\sigma$ is the representation $^g\sigma \in \widehat{N}$ defined by
\begin{equation}\label{act}
    ^g\sigma(n) = \sigma(g^{-1}ng),
\end{equation}
for all $n \in N$.
The \textbf{inertia subgroup} of $\sigma \in \widehat{N}$ is defined
\begin{equation*}
    I_G(\sigma) = \{g \in G : {}^g\sigma \sim \sigma\}.
\end{equation*}
\end{definition}
Note that ${}^g\sigma$ is irreducible, since any subspace invariant under ${}^g\sigma$ is also invariant under $\sigma$. Since $^{gh}\sigma(n) = \sigma((gh)^{-1}n(gh)) = \sigma(h^{-1}(g^{-1}ng)h) = {}^g(^h\sigma(n))$, \cref{act} defines an action of $G$ on $\widehat{N}$, and thus $I_G(\sigma)$ is the stabilizer of $\sigma$ in $G$. Notice that 
\begin{equation*}
    ^{m}\sigma(n) = \sigma(m^{-1}nm) = \sigma(m)^{-1}\sigma(n)\sigma(m)\hspace{3mm} \text{for } m,n\in N.
\end{equation*}
If $\chi$ is the character of $\sigma$ and $\chi_{m}$ is the character of $^{m}\sigma$, we have, for $n \in N$
\begin{equation*}
    \chi_{m}(n) = tr(^{m}\sigma(n)) = tr(\sigma(m)^{-1}\sigma(n)\sigma(m)) = tr(\sigma(n)) = \chi(n).
\end{equation*}
Thus, we have $^{m}\sigma \sim \sigma$ for $m \in N$, and so $N \leq I_{G}(\sigma)$.
\begin{lemma}\label{SubLemma}
If $\sigma$ and ${}^g\sigma$ are conjugate irreducible representations of a normal subgroup $N$ of a finite group $G$, and $I_G(\sigma)$ and $I_G({}^g\sigma)$ are the respective inertia subgroups, then $I_G(\sigma)$ and $I_G({}^g\sigma)$ are conjugate subgroups of $G$. In particular, $I_G(\sigma)$ and $I_G({}^g\sigma)$ are isomorphic.
\end{lemma}
We can now recall from \cite{Ceccherini-Silberstein2009} some central results of Clifford theory. Let $R$ be a family of coset representatives for the left $I_G(\sigma)$-cosets in $G$ with $e_G \in R$, that is
\begin{equation*}
    G = \bigsqcup_{r \in R}rI_G(\sigma).
\end{equation*}
Then $\{^g\sigma : g \in G\} = \{^r\sigma : r \in R\}$ and the representations $^r\sigma$ are pairwise inequivalent.
\begin{theorem}[\cite{Ceccherini-Silberstein2009} Theorem 2.1]\label{thm1}
Suppose that $N \unlhd G$ and $\sigma \in \widehat{N}$ and $\theta \in \widehat{G}(\sigma)$. If we set $d = [I_G(\sigma):N]$ and let $l$ denote the multiplicity of $\sigma$ in \emph{Res}$^G_N\theta$, we have:
\begin{enumerate}
    \item $\emph{Hom}_G(\emph{Ind}^G_N\sigma,\emph{Ind}^G_N\sigma) \cong \mathbb{C}^d$ as vector spaces, and
    \item $\emph{Res}_N^G \theta \cong l\bigoplus_{r \in R}{}^r\sigma$.
\end{enumerate}
\end{theorem}
The number $l = dim(\text{Hom}_N(\sigma,\text{Res}^G_N\theta))$ is called the \textbf{inertia index} of $\theta \in \widehat{G}(\sigma)$ with respect to $N$.
\begin{theorem}[Clifford Correspondence] \label{CliffCor}
Let $N \unlhd G$, $\sigma \in \widehat{N}$ and $I = I_G(\sigma)$, then
\begin{equation*}
    \widehat{I}(\sigma) \longrightarrow \widehat{G}(\sigma): \psi \longmapsto \emph{Ind}^G_I\psi
\end{equation*}
is a bijection. The inertia index of $\psi \in \widehat{I}(\sigma)$ with respect to $N$ coincides with the inertia index of $\emph{Ind}^G_I\psi$ with respect to $N$. In turn, the inertia index of $\emph{Ind}^G_I\psi$ with respect to $N$ is equal to the multiplicity $m_{\psi}$ of $\psi$ in $\emph{Ind}^I_N\psi$. Furthermore,
\begin{equation*}
    \emph{Res}^I_N\psi = m_{\psi}\sigma.
\end{equation*}
\end{theorem}
The following corollary of the Clifford correspondence is from section 8.1 of Serre's book\cite{serre_1977}.
\begin{corollary}\label{abelCorol}
If $N$ is an abelian normal subgroup of $G$, the degree of each irreducible representation $\rho$ of $G$ divides the index $[G:N]$ of $N$ in $G$.
\end{corollary}
Unfortunately, this correspondence does not tell us anything in the case where the inertia subgroup is all of $G$. The study of what happens in this case is known as stable Clifford theory and can be quite complicated \cite{craven_2019}. 
\begin{definition}
Let $\psi$ be a representation of $G/N$, the \textbf{inflation} $\widetilde{\psi}$ of $\psi$ is a representation of $G$ defined by setting 
\begin{equation*}
    \widetilde{\psi}(g) = \psi(gN) \hspace{3mm} \text{for all } g \in G.
\end{equation*}
If $\chi$ and $\widetilde{\chi}$ be characters of $\psi$ and $\widetilde{\psi}$ respectively, then the map $\chi \mapsto \widetilde{\chi}$ is a bijection between the irreducible characters of $G/N$ and the irreducible characters of $G$ with $N$ in their kernel (i.e. $\widetilde{\chi}(n)=\emph{deg}\;\widetilde{\chi}$). Note that deg$(\widetilde{\chi}) =$ deg$(\chi)$.
\end{definition}

Let $\sigma$ be an irreducible representation of $G$ and $\rho_1$ be the trivial representation of $N$ (the representation mapping every $n\in N$ to 1). If we suppose $\text{Res}^G_N\sigma \succeq \rho_1$ then notice that ${}^h\rho_1(g) = \rho_1(h^{-1}gh) = 1 = \rho_1(g)$ for all $h,g\in G$, thus ${}^h\rho_1 \sim \rho_1$ for all $h \in H$. Combining this observation with part \emph{2} of  \Cref{thm1} we see
\begin{equation*}
    \text{Res}_N^G \sigma \cong \bigoplus_{l = 1}^{\text{deg}(\sigma)}\rho_1.
\end{equation*}
So $N \leq \text{ker}(\sigma)$ and thus $\sigma$ is the inflation of an irreducible representation of $G/N$.
\begin{definition}
    Let $H\leq G$, and let $\sigma$ be a representation of $H$. We call a representation $\sigma'$ of $G$ an \textbf{extension} of $\sigma$ if $\emph{Res}_H^G\sigma' = \sigma$.
\end{definition}
We can now state a consequence of the Clifford correspondence that will prove very useful in our study of the Clifford group.
\begin{theorem}[The little group method; \cite{Ceccherini-Silberstein2009}, Theorem 5.1]\label{LittleTheorem}
    Let $G$ be a finite group with $N \unlhd G$ a normal subgroup. Suppose that any $\sigma\in\widehat{N}$ has an extension $\sigma'$ to its inertia group $I_G(\sigma)$. In $\widehat{N}$, define an equivalence relation $\approx$ by setting $\sigma_1 \approx \sigma_2$ if there exists $g\in G$ such that $^g\sigma_1 \sim \sigma_2$. Let $\Sigma$ be a set of representatives for the equivalence classes of $\approx$. For $\psi \in \widehat{I_G(\sigma)/N}$, let $\widetilde{\psi}$ be its inflation to $I_G(\sigma)$. Then 
    \begin{equation*}
        \widehat{G} = \{\emph{Ind}^G_{I_G(\sigma)}(\sigma'\otimes\widetilde{\psi}):\sigma \in \Sigma, \psi\in\widehat{I_G(\sigma)/N}\},
    \end{equation*}
    that is, the representations $\sigma'\otimes\widetilde{\psi}$ form a complete list of irreducible representations of $G$ and are pairwise inequivalent.
\end{theorem}
\begin{definition}
Let $Q$, $G$, and $N$ be groups. If we have an injective homomorphism $\iota: N \rightarrow G$, and a surjective homomorphism $\pi: G \rightarrow Q$, and if $\iota(N) = \text{ker}(\pi)$, then we call $G$ an \textbf{extension} of $Q$ by $N$. If $\iota(N)$ is contained in the center of $G$, then we call $G$ a \textbf{central extension}. A group extension $G$ is often written as a short exact sequence
\begin{equation*}
    1\rightarrow N\xrightarrow{\iota} G\xrightarrow{\pi} Q\rightarrow 1.
\end{equation*}
\end{definition}
\Cref{LittleTheorem} classifies the irreducible representations of a group extension $G$ under the constraint that the irreducible representations $\sigma$ of the normal subgroup $N$ can always be extended to representations $\sigma'$ of the corresponding inertia subgroup $I_G(\sigma)$. 
 
Let $1 \rightarrow B \xrightarrow[]{\iota} G \xrightarrow[]{\pi} H \rightarrow 1$ be a central extension. When $G$ is a central extension and $G \not\cong H\times B$, the little group method does not apply. To examine this case we require more specialized machinery. A \textbf{section} of the extension is a map $t:H\rightarrow G$ which is a right inverse for $\pi$, that is
\begin{equation*}
    \pi(t(h)) = h
\end{equation*}
for all $h \in H$. We call the section \textbf{normalized} if $t(e_H) = e_G$. For $h,k \in H$ we have 
\begin{equation*}
    \pi[t(h)t(k)] = \pi(t(h))\pi(t(k)) = hk = \pi(t(hk)),
\end{equation*}
so there exists a unique $b(h,k) \in B$ such that 
\begin{equation*}
    t(h)t(k) = t(hk)\iota[b(h,k)].
\end{equation*}
Let $\widehat{H}^{\alpha}$ denote all the irreducible projective representations of a finite group $H$ with factor set $\alpha$. We may now state a version of the little group method for central extensions.
\begin{proposition}[\cite{ceccherini2022representation}, Proposition 7.24]\label{centerprop}
For every $\xi \in \widehat{B}$ we have $I_G(\xi) = G$. Let
\begin{equation*}
    \eta(h,k) = \xi(b(h,k)), 
\end{equation*} 
Let $\bar{\eta}(h,k) = (\eta(h,k))^{-1}$, then the map 
\begin{equation*}
    \widehat{H}^{\bar{\eta}} \longrightarrow \widehat{G}(\xi):
    \Phi \longmapsto\Theta
\end{equation*}
is a bijection, with $\Theta$ defined by 
\begin{equation}\label{thetaeq}
    \Theta(t(h)b) = \xi(b)\Phi(h)
\end{equation}
for all $h \in H(\xi)$, $b \in B$. Finally,
\begin{equation*}
    \widehat{G} = \left\{ \Theta: \xi \in \widehat{B}, \Theta \text{ as in \emph{\cref{thetaeq}}}, \Phi \in \widehat{I_G(\xi)/B}^{\bar{\eta}}\right\}.
\end{equation*}
\end{proposition}
\section{The Pauli and Clifford groups}
\subsection{Definitions}
Here we recall the definitions of the Pauli and Clifford groups.
\begin{definition}
Let $U(d)$ be the set of d-by-d unitary matrices, where d is some power of 2. This has a standard representation on $\mathbb{C}^d$, the complex vector space \cite{Fulton1991RepresentationTA}. Let $v_0, v_1$ be an orthonormal basis of $\mathbb{C}^2$ and define the linear operators $X$, $Y$ and $Z$ by
\begin{equation*}
    Xv_l = v_{l+1}, \hspace{4mm} Zv_l = (-1)^lv_l, \hspace{4mm} Yv_l = -iZXv_l = (-1)^{l}iv_{l+1}
\end{equation*}
for $l \in \{0,1\}$, with addition over indices being modulo 2. These operators are unitary. We define the \textbf{n-qubit Pauli group} $\mathcal{P}_n$ as the subset of the unitary group $U(2^n)$ consisting of all $n$-fold tensor products of elements of $\mathcal{P}_1 := \left<X,Z,iI_2\right>$, where $I_n$ is the identity on $\mathbb{C}^{n}$.
\end{definition}
$\mathcal{P}_1$ is a group of order $16$ with centre $|Z(\mathcal{P}_1)| = 4$. Since $\mathcal{P}_n$ consists of $n$-fold tensor products of elements of $\mathcal{P}_1$ it is a central product of copies of $\mathcal{P}_1$, and thus $|\mathcal{P}_n| = 4^{n+1}$. The operators $X$, $Y$, and $Z$ can be written in matrix form with respect to the eigenbasis of $Z$ as
\begin{equation*}
    X = \begin{bmatrix}0&1\\
                       1&0\end{bmatrix}
    \hspace{4mm}Z =\begin{bmatrix}1&0\\
                                  0&-1\end{bmatrix}
    \hspace{4mm}Y = \begin{bmatrix}0&-i\\
                                   i&0\end{bmatrix}.
\end{equation*}
These are known as the Pauli matrices. 
\begin{definition}
The \textbf{n-qubit Clifford group} $\emph{Cliff}(n)$ is the normalizer of the $n$-qubit Pauli group in the unitary group
\begin{equation*}
    \emph{Cliff}(n) = \left\{U\in U(2^n) : U\mathcal{P}_nU^\dagger\subseteq\mathcal{P}_n\right\}.
\end{equation*}
Since in quantum information theory global phases have no effect on measurement outcomes, it is common to define the Clifford group modulo phases. We denote this group
\begin{equation*}
    \mathcal{C}_n = \left\{U\in U(2^n) : U\mathcal{P}_nU^\dagger\subseteq\mathcal{P}_n\right\}/U(1),
\end{equation*}
and will call it the \textbf{n-qubit projective Clifford group} to differentiate it from other ways the Clifford group is defined in the literature.
\end{definition}
We would like to understand the representation theory of the projective Clifford group.
 
The Clifford group $\text{Cliff}(n)$ is generated by the Hadamard ($H$) and Phase ($S$) gates on each qubit (i.e. on each tensor factor), and Controlled-Z ($CZ$) gate on each pair of qubits, along with phases. In matrix form these gates are
\begin{equation*}
    H = \frac{1}{\sqrt{2}}\begin{bmatrix}1&1\\
                                        1&-1\end{bmatrix}
    \hspace{4mm}S =\begin{bmatrix}1&0\\
                                  0&i\end{bmatrix}
    \hspace{4mm}CZ =\begin{bmatrix}1&0&0&0\\
                                  0&1&0&0\\
                                  0&0&1&0\\
                                  0&0&0&-1\end{bmatrix}.
\end{equation*}
Operators that are $n$-fold tensor products only of $2$-by-$2$ matrices are said to consist only of single-qubit operations. If an operator has $4$-by-$4$ matrices in its tensor decomposition, such as $CZ$ that do not decompose further into tensor factors, it is said to contain multi-qubit operations.
 
Multi-qubit Pauli operators either commute or anticommute. Notice that the group $C_4 = \langle iI_n\rangle$ of phases in $\mathcal{P}_n$ is the centre of $\mathcal{P}_n$. 
\begin{definition}
Define the \textbf{n-qubit projective Pauli group} to be $\widetilde{\mathcal{P}}_n = \mathcal{P}_n/C_4$. Since $C_4$ contains the commutator subgroup $C_2 = \langle -I_n\rangle$ of $\mathcal{P}_n$, we have that $\widetilde{\mathcal{P}}_n$ is abelian.
\end{definition}
Notice, $\widetilde{\mathcal{P}}_n$ is a normal subgroup of $\mathcal{C}_n$. Since $\widetilde{\mathcal{P}}_n \cong \mathbb{Z}^{2n}_2$ we will often just write $\mathbb{Z}^{2n}_2$ for the projective Pauli group. 
\subsection{The symplectic structure of the Clifford group}\label{proofsection}
The quotient of the Clifford group by the Pauli group and phases is essential to the stable Clifford theory of the Clifford group. We will need the following proposition, which we present without proof.
\begin{proposition}[\cite{Selinger_2015}, Prop 3.3]\label{prop3}
Let $\phi:\mathcal{P}_n \rightarrow \mathcal{P}_n$ be an automorphism of the Pauli group that fixes scalars. That is, $\phi(i^lI_{2^n}) = i^lI_{2^n}$. Then there exists $U \in \emph{Cliff}(n)$ unique up to phases such that for all $P \in \mathcal{P}_n$ we have $UPU^{\dagger} = \phi(P)$.
\end{proposition}
\begin{remark}
This says that $\mathcal{C}_n = \emph{Aut}_{\langle i\rangle}(\mathcal{P}_n)$, that is, the Clifford group consists of the automorphisms of the Pauli group that fix the centre.
\end{remark}
Following arguments from \cite{Gross2017} and \cite{debeaudrap2012linearized} we can show the following.
\begin{theorem}\label{QuotThm}
The quotient of the Clifford group $\emph{Cliff}(n)$ by the Pauli group and phases is
\begin{equation*}
    \mathcal{C}_n/\widetilde{\mathcal{P}}_n \cong Sp(2n,2),
\end{equation*}
the symplectic group of degree $2n$ over $\mathbb{Z}_2$.
\end{theorem}
\begin{proof}
For $\mathbf{x} = (\mathbf{p},\mathbf{q}) \in \mathbb{Z}^{2n}$ define the \textbf{Weyl Operator}
\begin{equation*}
    W_\mathbf{x} = W_{\mathbf{p},\mathbf{q}} = i^{-\mathbf{p}\cdot 
    \mathbf{q}}(Z^{p_1}X^{q_1})\otimes\cdots\otimes(Z^{p_n}X^{q_n}).
\end{equation*}
Clearly, all Weyl operators are elements of the Pauli group $\mathcal{P}_n$, and any element of the Pauli group is a Weyl operator up to a factor of $i$ to some power. Weyl operators only depend on $\mathbf{x}$ modulo $4$, since
\begin{equation}\label{Weyl1}
    W_{\mathbf{x}+2\mathbf{z}} = (-1)^{[\mathbf{x},\mathbf{z}]}W_{\mathbf{x}},
\end{equation}
where we have introduced the $\mathbb{Z}$-valued symplectic form $[\cdot,\cdot]$ on $\mathbb{Z}^{2n}$
\begin{equation*}
    [\mathbf{x},\mathbf{z}] = [(\mathbf{p},\mathbf{q}),(\mathbf{p}',\mathbf{q}')] = \mathbf{p}\cdot\mathbf{q}'-\mathbf{q}\cdot\mathbf{p}'.
\end{equation*}
We will use this form when $\mathbf{x},\mathbf{y}\in\mathbb{Z}^{2n}_4$, and interpret accordingly. For example, by direct computation we have
\begin{equation*}
    W_{\mathbf{x}}W_{\mathbf{y}} = i^{[\mathbf{x},\mathbf{y}]}W_{\mathbf{x}+\mathbf{y}}.
\end{equation*}
Then for all $\mathbf{x},\mathbf{y}\in \mathbb{Z}^{2n}_4$, we have
\begin{equation*}
    W_{\mathbf{x}}W_{\mathbf{y}} = i^{[\mathbf{x},\mathbf{y}]}W_{\mathbf{x}+\mathbf{y}} =  i^{-[\mathbf{y},\mathbf{x}]}W_{\mathbf{y}+\mathbf{x}} = i^{-2[\mathbf{y},\mathbf{x}]}W_{\mathbf{y}}W_{\mathbf{x}} = (-1)^{[\mathbf{x},\mathbf{y}]}W_{\mathbf{y}}W_{\mathbf{x}}.
\end{equation*}
Thus the commutation relation depends on $[\mathbf{x},\mathbf{y}]$ mod $2$. By definition if $U \in \text{Cliff}(n)$ then for every $\mathbf{x} \in \mathbb{Z}^{2n}_2$, $UW_{\mathbf{x}}U^{\dagger}$ is proportional to a Weyl operator  $W_{\mathbf{x}'}$ by some power of $i$ and by \cref{Weyl1} we can take $\mathbf{x}' \in \mathbb{Z}^{2n}_2$. We define the function $g: \mathbb{Z}_2^{2n} \rightarrow \mathbb{Z}_4$ where $i^{g(\mathbf{x})}W_{\mathbf{x}'} = UW_{\mathbf{x}}U^{\dagger}$. Since conjugation preserves commutation relations, we have 
\begin{equation*}
    (-1)^{[\mathbf{x},\mathbf{y}]}W_{\mathbf{y}'}W_{\mathbf{x}'} = W_{\mathbf{x}'}W_{\mathbf{y}'} = (-1)^{[\mathbf{x}',\mathbf{y}']}W_{\mathbf{y}'}W_{\mathbf{x}'}.
\end{equation*}
Thus the map $\Gamma: \mathbf{x} \mapsto \mathbf{x}'$ preserves the symplectic form $[\cdot,\cdot]$. Furthermore,
\begin{equation*}
\begin{split}
    i^{g(\mathbf{x}+\mathbf{y})+[\mathbf{x},\mathbf{y}]}W_{(\mathbf{x}+\mathbf{y})'} &= Ui^{[\mathbf{x},\mathbf{y}]}W_{\mathbf{x}+\mathbf{y}}U^{\dagger} = UW_{\mathbf{x}}W_{\mathbf{y}}U^{\dagger}
    \\
    &= UW_{\mathbf{x}}U^{\dagger}UW_{\mathbf{y}}U^{\dagger} = i^{g(\mathbf{x})+g(\mathbf{y})}W_{\mathbf{x}'}W_{\mathbf{y}'} 
    \\
    &= i^{[\mathbf{x},\mathbf{y}]+g(\mathbf{x})+g(\mathbf{y})}W_{\mathbf{x}'+\mathbf{y}'}.
\end{split}
\end{equation*}
Thus 
\begin{equation*}
    i^{g(\mathbf{x}+\mathbf{y})}W_{(\mathbf{x}+\mathbf{y})'} = i^{g(\mathbf{x})+g(\mathbf{y})}W_{\mathbf{x}'+\mathbf{y}'}.
\end{equation*}
Then
\begin{equation*}
    W_{(\mathbf{x}+\mathbf{y})'} = \pm W_{\mathbf{x}'+\mathbf{y}'}.
\end{equation*}
So by \cref{Weyl1}, $\Gamma$ is compatible with addition in $\mathbb{Z}^{2n}_2$. Since $\mathbb{Z}_2$ has only the scalars $0$ and $1$, we deduce that $\Gamma$ is linear and thus an element of the symplectic group $Sp(2n,2)$. Then for each $U \in \text{Cliff}(n)$ there is a $\Gamma \in Sp(2n,2)$ and a function $g: \mathbb{Z}_2^{2n} \rightarrow \mathbb{Z}_4$ such that 
\begin{equation*}
    UW_{\mathbf{x}}U^{\dagger} = i^{g(\mathbf{x})}W_{\Gamma(\mathbf{x})}.
\end{equation*}
Now notice that the $n$-qubit Pauli matrices form a basis of the vector space $M_{2^n}(\mathbb{C})$ of all $2^n$-by-$2^n$ matrices. If we specify the action of $U \in \text{Cliff}(n)$ on a generating set of the $\mathcal{P}_n$, then we determine $U$ up to a phase since $U' = e^{i\theta}U$ has the same action as $U$ by conjugation. From \Cref{prop3} we know that for any scalar fixing automorphism $\phi$ of $\mathcal{P}_n$ there exists some $U \in \text{Cliff}(n)$ such that
\begin{equation*}
    \phi(P) = UPU^{\dagger}
\end{equation*}
for all $P \in \mathcal{P}_n$. 
 
For any linear $\Gamma: \mathbb{Z}^{2n}_4 \rightarrow  \mathbb{Z}^{2n}_4$ that preserves the symplectic product modulo $4$, we can define the map $\Phi : \mathcal{P}_n \rightarrow \mathcal{P}_n$ by
\begin{equation*}
    \Phi(W_{\mathbf{x}}) = W_{\Gamma\mathbf{x}}, \hspace{3mm} \text{for all } \; \mathbf{x} \in \mathbb{Z}_4^{2n}.
\end{equation*}
To see this is well defined, notice that $W_{\Gamma\mathbf{x}}$ is expressible as a linear combination of other Weyl operators only if $\pm W_{\Gamma\mathbf{y}} = W_{\Gamma\mathbf{x}}$ for some $\mathbf{y} \in \mathbb{Z}_4^{2n}$. Then by dimension counting and \cref{Weyl1} we have $\Gamma\mathbf{y} = \Gamma(\mathbf{x}+2\mathbf{z}) = \Gamma\mathbf{x}+2\Gamma\mathbf{z}$. Thus the sign is given by $(-1)^{[\Gamma\mathbf{x},\Gamma\mathbf{z}]} = (-1)^{[\mathbf{x},\mathbf{z}]}$, so $\Phi$ is well defined. Furthermore, since
\begin{equation*}
    \Phi(W_{\mathbf{x}})\Phi(W_{\mathbf{y}}) = W_{\Gamma\mathbf{x}}W_{\Gamma\mathbf{y}} = i^{[\mathbf{x},\mathbf{y}]}W_{\Gamma\mathbf{x}+\Gamma\mathbf{y}}
    = i^{[\mathbf{x},\mathbf{y}]}\Phi(W_{\mathbf{x}+\mathbf{y}}) = \Phi(W_{\mathbf{x}}W_{\mathbf{y}}),
\end{equation*}
extending $\Phi$ by linearity defines an automorphism on $\mathcal{P}_n$ that fixes scalars. We thus have a $U \in \text{Cliff}(n)$ such that 
\begin{equation*}
    UW_{\mathbf{x}}U^{\dagger} = W_{\Gamma\mathbf{x}}.
\end{equation*}
Let $\{\mathbf{e}_j\}_{j=1}^{2n}$ be a basis of $\mathbb{Z}_2^{2n}$. Fix $\Gamma \in Sp(2n,2)$, and let 
\begin{equation*}
\Gamma \mathbf{e}_j = \mathbf{v}_j \;\text{for all }\;j \in\{1,\dots,2n\}.
\end{equation*}
In other words, let $C = [\mathbf{v}_1\cdots\mathbf{v}_{2n}]$ be the matrix corresponding to the symplectic map $\Gamma$. Define $\overline{\mathbf{v}}_1 := \mathbf{v}_1$, and for each subsequent $j>1$ define $\overline{\mathbf{v}}_j := \mathbf{v}_j+2\mathbf{x}_j$ with $\mathbf{x}_j \in \mathbb{Z}_2^{2n}$ chosen such that
\begin{equation*}
[\mathbf{v}_j,\overline{\mathbf{v}}_j] = 0 \;\text{mod}\;4\;\; \text{ and }\;\;                   [\overline{\mathbf{v}}_h,\overline{\mathbf{v}}_j] = \delta_{h,n+j}-\delta_{j,n+h} \;\text{mod}\;4\; \text{ for }h<j,
\end{equation*}
where $\delta_{a,b}$ is the Kronecker delta. Notice that for all $j$ we have $\overline{\mathbf{v}}_j \in \mathbb{Z}_4^{2n}$. Since $\Gamma$ preserves the symplectic product mod $2$ (and thus preserves commutators and anticommutators), we have both of these restrictions already satisfied mod $2$. The matrix $\overline{C} = [\overline{\mathbf{v}}_1\cdots\overline{\mathbf{v}}_{2n}]$ is symplectic modulo $4$ and $\overline{C} = C$ mod $2$.
This means that for each $\Gamma \in Sp(2n,2)$ there is a $\widetilde{\Gamma} \in Sp(\mathbb{Z}_4^{2n})$ such that $\Gamma \mathbf{x} = \widetilde{\Gamma}\mathbf{x}$ mod $2$ thus we obtain $ \widetilde{\Gamma}\mathbf{x}-\Gamma \mathbf{x} = 2\mathbf{z}$ for some $\mathbf{z} \in \mathbb{Z}^{2n}$. If we define the function $f:\mathbb{Z}_2^{2n} \rightarrow \mathbb{Z}_2$ by $f(\mathbf{x}) = [\Gamma\mathbf{x},\mathbf{z}]$ mod $2$, we have
\begin{equation*}
    (-1)^{f(\mathbf{x})}W_{\Gamma\mathbf{x}} = W_{\Gamma\mathbf{x}+2\mathbf{z}} = W_{\widetilde{\Gamma}\mathbf{x}}.
\end{equation*}
This implies that for every $\Gamma \in Sp(2n,2)$ there exists a $U \in \text{Cliff}(n)$ and a function $f: \mathbb{Z}_2^{2n} \rightarrow \mathbb{Z}_2$ such that for all $\mathbf{x} \in \mathbb{Z}_2^{2n}$
\begin{equation} \label{Cliffaction}
    UW_{\mathbf{x}}U^{\dagger} = (-1)^{f(\mathbf{x})}W_{\Gamma\mathbf{x}}.   
\end{equation}
This $U$ is determined uniquely up to phase since we have determined its action by conjugation. Thus we have a surjective correspondence $U \mapsto \Gamma$ between $\mathcal{C}_n$ and $Sp(2n,2)$, and we see that the quotient of $\text{Cliff}(n)$ by Paulis and phases is $Sp(2n,2)$.
\end{proof}
Note that this implies that $\mathcal{C}_n$ is an extension of $Sp(2n,2)$ by $\mathbb{Z}_2^{2n}$, but since we cannot specify that $f \equiv 0$ for all choices of $U \in$ Cliff($n$) in \cref{Cliffaction} the extension does not split for $n > 1$. For $n = 1$ we have $\mathcal{C}_1 \cong S_4 \cong Sp(2,2)\ltimes \mathbb{Z}^2_2$.
Since $|Sp(2n,2)| = 2^{n^2}\prod_{j=1}^n(2^{2j}-1)$, we obtain the following immediate corollary.
\begin{corollary}
The order of the Clifford group is 
\begin{equation*}
    |\mathcal{C}_n| = |\widetilde{\mathcal{P}}_n||Sp(2n,2)| = 2^{n^2+2n}\prod_{j=1}^n(2^{2j}-1).
\end{equation*}
\end{corollary}
 
\subsection{The character table of the projective Pauli group}\label{sec:paulirep}
Since the $n$-qubit projective Pauli group is abelian it has only degree one irreducible characters. One-dimensional representations and characters coincide since the trace leaves 1-by-1 matrices invariant. Elements of the $n$-qubit projective Pauli group have order at most 2. Thus, the character of an element of $\widetilde{\mathcal{P}}_n$ must be $\pm1$. $\widetilde{\mathcal{P}}_n$ is generated by $\{[X_1],[Z_1],\dots,[X_n],[Z_n]\}$, where $[A_j]$ is the equivalence class of the Pauli operator $A$ acting on the $j^{\text{th}}$ qubit
\begin{equation*}
    A_j = I_2^{\otimes j-1}\otimes A\otimes I_2^{\otimes n-j}.
\end{equation*}
So a character of $\widetilde{\mathcal{P}}_n$ is fully determined by a choice of $\pm1$ for $[X_i]$ and $[Z_i]$ for each $i \in \{1,\dots,n\}$. The $4$ choices for each qubit leave us with $4^n$ choices for the whole group. There are $4^n = |\widetilde{\mathcal{P}}_n|$ characters since $\widetilde{\mathcal{P}}_n$ is abelian and thus has only singleton conjugacy classes. Irreducible characters that disagree on any one element must be distinct, so this completely determines the character table of $\widetilde{\mathcal{P}}_n$.
Thus, the character table of $\widetilde{\mathcal{P}}_n$ can be written by filling the first row and column of a $4^n$-by-$4^n$ table with ones, then in the rest of each remaining row writing each permutation of $\frac{4^n}{2}-1$ ones and $\frac{4^n}{2}$ negative ones.
\section{The inertia subgroup}\label{Section:Inert}
To begin our study of the character theory of the $n$-qubit projective Clifford group, we examine the inertia subgroups of the representations of the $n$-qubit projective Pauli group in the $n$-qubit projective Clifford group.
\begin{lemma}\label{ConjLem}
Let $\sigma$ and $\rho$ be nontrivial irreducible representations of $\widetilde{\mathcal{P}}_n$, then there exists $g \in \mathcal{C}_n$ such that ${}^g\sigma \sim \rho$. In other words, all nontrivial irreducible representations of $\widetilde{\mathcal{P}}_n$ are conjugate in $\mathcal{C}_n$.
\end{lemma}
\begin{proof}
We begin the proof by noticing that 
\begin{equation}\label{Conj1}
\begin{split}
    HXH^{-1} &= Z
                                             \\
    HZH^{-1} &= X
                                             \\
    HYH^{-1} &= -Y.
\end{split}
\end{equation}
So we have that conjugation of Pauli matrices by $H$ maps $X$ to $Z$ and vice versa, while mapping $Y$ to $-Y$. Thus conjugation by $[H]$ maps $[X]$ to $[Z]$ and vice versa, while leaving $[Y]$ invariant. We can calculate
\begin{equation}\label{Conj2}
\begin{split}
    SXS^{-1} &= Y
                                             \\
    SZS^{-1} &= Z
                                             \\
    SYS^{-1} &= -X,
\end{split}
\end{equation}
thus conjugation by $[S]$ maps $[X]$ to $[Y]$ and vice versa, while leaving $[Z]$ invariant. Furthermore conjugation by $[H][S][H]$ maps $[Z]$ to $[Y]$ and vice versa, while leaving $[X]$ invariant.
We see that we can permute the non-identity elements of the one-qubit projective Pauli group in any way via conjugation by elements of $\mathcal{C}_n$.
 
We now turn our attention to 2-qubit operators. Consider the swap gate, if $A$ and $B$ are any 2-by-2 matrices we have
\begin{equation*}
        (SWAP)(A\otimes B)(SWAP) = B\otimes A.
\end{equation*}
Our previous calculations for 1-qubit matrices tell us that any pair of nontrivial representations $\sigma$ and $\rho$ of $\widetilde{\mathcal{P}}_n$ that have the same number of pairs of generators ($[X_i],[Z_i]$) in their kernels, that is
\begin{equation*}
    |\{i\in\{1,\dots,n\}: \rho([X_i]) = \rho([Z_i]) = 1\}| = |\{i\in\{1,\dots,n\}: \sigma([X_i]) = \sigma([Z_i]) = 1\}|,
\end{equation*}
are conjugate in $\mathcal{C}_n$. Consider the two representations $\rho$ and $\sigma$ of $\widetilde{\mathcal{P}}_2$ defined by 
\begin{equation*}
    \begin{split}
        \sigma([X\otimes I]) &=  \sigma([I\otimes Z]) = -1
        \\
        \sigma([Z\otimes I]) &= \sigma([I\otimes X]) = 1
                                             \\
   \rho([X\otimes I]) &=  \rho([I\otimes X]) = 
   \rho([Z\otimes I]) = 1 
   \\
   \rho([I\otimes Z]) &= -1.
    \end{split}
\end{equation*}
Now we calculate
\begin{equation*}
    \begin{split}
        CZ(I\otimes X)CZ &= (Z\otimes X)
                            \\
        CZ(Z\otimes I)CZ &= (Z\otimes I)
                            \\
CZ(X\otimes I)CZ &= (X\otimes Z)
                            \\
CZ(I\otimes Z)CZ &= (I\otimes Z).
    \end{split}
\end{equation*}
Thus we have
\begin{equation*}
\begin{split}
    {}^{CZ}\rho([X\otimes I]) &= \rho([X\otimes Z]) = -1 = \sigma([X\otimes I])
    \\
    {}^{CZ}\rho([Z\otimes I]) &= \rho([Z\otimes I]) = 1 = \sigma([Z\otimes I])
    \\
    {}^{CZ}\rho([I\otimes X]) &= \rho([Z\otimes X]) = 1 = \sigma([I\otimes X])
    \\
    {}^{CZ}\rho([I\otimes Z]) &= \rho([I\otimes Z]) = -1 = \sigma([I\otimes Z]).
\end{split}
\end{equation*}
Thus nontrivial irreducible representations $\sigma$ and $\rho$ of $\widetilde{\mathcal{P}}_2$ with differing numbers of $([X_i],[Y_i])$ pairs in their kernels, that is
\begin{equation*}
    |\{i\in\{1,2\}: \rho([X_i]) = \rho([Z_i]) = 1\}| \neq |\{i\in\{1,2\}: \sigma([X_i]) = \sigma([Z_i]) = 1\}|,
\end{equation*}
are conjugate in $\mathcal{C}_2$. Since restricting irreducible representations of $\widetilde{\mathcal{P}}_n$ to any two qubits gives an irreducible representation of $\widetilde{\mathcal{P}}_2$, taking all the previous calculations together, we have that all nontrivial irreducible representations of $\widetilde{\mathcal{P}}_n$ are conjugate in $\mathcal{C}_n$.
\end{proof}
Since, by \Cref{SubLemma}, conjugate representations of normal subgroups have isomorphic inertia subgroups, we see that there is only one inertia subgroup to calculate for the nontrivial representations of the projective Pauli group in the Clifford group. We have the following immediate corollary.
\begin{corollary}\label{casescorol}
 If $\rho$ is an irreducible representation of $\mathcal{C}_n$ and $\sigma$ an irreducible representation of $\widetilde{\mathcal{P}}_n$ with $\emph{Res}^{\mathcal{C}_n}_{\widetilde{\mathcal{P}}_n} \rho \succeq \sigma$ then one of two cases holds:
\begin{enumerate}
    \item $\sigma$ is trivial, and $\widetilde{\mathcal{P}}_n$ is in the kernel of $\rho$. In this case, $\rho$ is the inflation of an irreducible representation of $Sp(2n,2)$, or
    \item $\sigma$ is nontrivial in which case we can apply \Cref{ConjLem} and \Cref{thm1} to obtain 
    \begin{equation} \label{decomp}
        \text{Res}^{\mathcal{C}_n}_{\widetilde{\mathcal{P}}_n}\rho = l\bigoplus_{\mathclap{\substack{\theta \in \text{Irr}({\widetilde{\mathcal{P}}_n})\\
        \theta \text{ nontrivial}}}}\theta,
    \end{equation}
where $\text{Irr}(\widetilde{\mathcal{P}}_n)$ is the set of irreducible representations of $\widetilde{\mathcal{P}}_n$ and $l$ is the inertia index of $\rho$ with respect to $\widetilde{\mathcal{P}}_n$.
\end{enumerate}
Additionally, since $\widetilde{\mathcal{P}}_n$ is an abelian normal subgroup of $\mathcal{C}_n$ we have that the degree of $\rho$ divides 
\begin{equation*}
    [\mathcal{C}_n:\widetilde{\mathcal{P}}_n] = 2^{n^2}\prod_{j=1}^n(2^{2j}-1),
\end{equation*} 
by \Cref{abelCorol}.
\end{corollary}
If we specialize to case 2, then \cref{decomp} and the fact that all irreducible representations of $\widetilde{\mathcal{P}}_n$ have degree $1$ imply that the degree of $\rho$ is divisible by $4^n-1$ (the number of nontrivial irreducible representations of $\widetilde{\mathcal{P}}_n$). If $\chi$ is the character of $\rho$, then \cref{decomp} implies
\begin{equation*}
    \text{Res}^{\mathcal{C}_n}_{\widetilde{\mathcal{P}}_n}\chi = l\sum_{\mathclap{\substack{\psi \in \text{IrrChar}({\widetilde{\mathcal{P}}_n})\\
        \psi \text{ nontrivial}}}}\psi,
\end{equation*}
where $\text{IrrChar}(\widetilde{\mathcal{P}}_n)$ is the set of irreducible characters of $\widetilde{\mathcal{P}}_n$. In particular, if $g \in \widetilde{\mathcal{P}}_n$ is a non-identity element then $\chi(g) = -l$, since for any such $g$ the summand takes the value $-1$ a total of $2^{2(n-1)+1}$ times and takes the value $1$ a total of $2^{2(n-1)+1}-1$ times.
 
To understand case 2, we need to calculate the inertia subgroup $I_{\mathcal{C}_n}(\sigma)$ of a nontrivial representation $\sigma$ of the Pauli group in the Clifford group. For a two qubit gate $A$, let $A_{i,j}$ denote $A$ acting on the pair of qubits $i$ and $j$. For $M = CX_{1,2}(Z_1H_1X_2)CX_{1,2}$, we have the following theorem.
\begin{theorem}\label{InertThm}
For $n \geq 2$ the inertia subgroup of a nontrivial representation of $\widetilde{\mathcal{P}}_n$ in $\mathcal{C}_n$ is isomorphic to $IN_n:=\langle \{[M],[H_1],[X_1],[I\otimes A]:\text{for } A\in \emph{Cliff}(n-1)\}\rangle$.
\end{theorem}
\begin{proof}
Notice that if $\sigma$ is a nontrivial irreducible representation of $\widetilde{\mathcal{P}}_n$, and $\psi$ an irreducible representation of $I = I_{\mathcal{C}_n}(\sigma)$ with $\text{Res}^{I}_{\widetilde{\mathcal{P}}_n}\psi\succeq\sigma$ then by the Clifford correspondence we have
\begin{equation*}
    m_{\psi}(2^{2n}-1) = \text{deg }m_{\psi}\bigoplus_{\mathclap{\substack{\theta \in \text{Irr}(\widetilde{\mathcal{P}}_n)\\\theta \text{ nontrivial}}}} \theta = \text{deg Ind}^{\mathcal{C}_n}_{I}\psi=[\mathcal{C}_n:I]\text{deg }\psi,
\end{equation*}
where $m_{\psi}$ is the inertia index of $\psi$ with respect to $\widetilde{\mathcal{P}}_n$. Additionally, by the Clifford correspondence,
\begin{equation*}
    [\mathcal{C}_n:I]\text{deg }\psi = [\mathcal{C}_n:I]m_{\psi}\text{deg }\sigma = m_{\psi}[\mathcal{C}_n:I],
\end{equation*}
Thus $[\mathcal{C}_n:I] = 2^{2n}-1$ and 
\begin{equation*}
    |I| = \frac{1}{2^{2n}-1}|\mathcal{C}_n| = 2^{n^2+2n}\prod_{j=1}^{n-1}(2^{2j}-1) = 2^{2n+1}|\mathcal{C}_{n-1}|.
\end{equation*}
So if for any particular $\sigma$ we can find a subgroup of $\mathcal{C}_n$ that preserves $\sigma$ under conjugation and has this order, then we have found the inertia subgroup.
 
Consider the irreducible character $\sigma_1$ of $\widetilde{\mathcal{P}}_n$ defined by $\sigma_1([X_1]) = \sigma_1([Z_1]) = -1$ and $\sigma_1([X_i]) = \sigma_1([Z_i]) = 1$ for all $i \in \{2,\dots,n\}$. We want to calculate $I_{\mathcal{C}_n}(\sigma_1) = \{g \in \mathcal{C}_n : {}^g\sigma_1 \sim \sigma_1\}$. So we want to find the elements of $\mathcal{C}_n$ that preserve the presence of $X$ or $Z$ in the first tensor factor by conjugation. We immediately see that $[I\otimes A] \in I_{\mathcal{C}_n}(\sigma_1)$ for $A\in \text{Cliff}(n)$ since operations restricted to other qubits do not affect the first qubit. Similarly, conjugation by Pauli elements preserves $\sigma_1$. Since conjugation by $H$ simply exchanges $X$ and $Z$, we also have $[H_1] \in I_{\mathcal{C}_n}(\sigma_1)$. Additionally, we have the operator
\begin{equation*}
    CX(ZH\otimes X)CX = \frac{1}{\sqrt{2}}\begin{bmatrix}
                            0&1&1&0\\               
                            1&0&0&1\\
                            -1&0&0&1\\       
                            0&-1&1&0\end{bmatrix}.
\end{equation*}
The action of this matrix on $\mathcal{P}_2$ by conjugation is
\begin{equation*}
\begin{split}
    CX(ZH\otimes X)CX(X\otimes I)CX(HZ\otimes X)CX &= -Z\otimes X
    \\
    CX(ZH\otimes X)CX(Z\otimes I)CX(HZ\otimes X)CX &= X\otimes X
    \\
    CX(ZH\otimes X)CX(I\otimes X)CX(HZ\otimes X)CX &= I\otimes X
    \\
    CX(ZH\otimes X)CX(I\otimes Z)CX(HZ\otimes X)CX &= -ZX\otimes ZX.
\end{split}
\end{equation*}
Notice that $\sigma_1([ZX\otimes ZX]) = \sigma_1([Z\otimes Z])\sigma_1([X\otimes X]) = (-1)^2 = \sigma_1([I\otimes Z])$. Thus the action of $[CX(ZH\otimes X)CX]$ preserves $\sigma_1$. The conjugation action of $IN_n$ leaves $[X_1Z_1]$ invariant, thus there are $2^{2n-1}$ possible images of the pair $([X_1],[Z_1])$. The order of $IN_n$ is thus 
\begin{equation*}
    (2^{2n-1})(2^{2n})|Sp(2(n-1),2)| = 2^{2n+1}|\mathcal{C}_{n-1}|.
\end{equation*}
\end{proof}
Note that there is no way to write $[X_1]$ in terms of the other generators of $IN_n$. Since $[X_1]^{-1} = [X_1]$, any reduction we perform on a word written in these generators preserves the parity of the number of $[X_1]$s. For $g$ a word in $IN_n$, let $n_{X_1}(g)$ be the number of times $[X_1]$ appears in $g$. The above analysis implies that the map 
\begin{equation*}
\begin{split}
    \sigma_1': IN_n &\longrightarrow \{1,-1\}
    \\
    g &\longmapsto (-1)^{n_{X_1}(g)}
\end{split}
\end{equation*}
is an irreducible character of $IN_n$. Furthermore, we have that $\text{Res}^{IN_n}_{\mathcal{P}_n}\sigma_1' = \sigma_1$. Thus $\sigma_1'$ is an extension of $\sigma_1$ to $IN_n$, its own inertia subgroup. Since all nontrivial irreducible representations of the projective Pauli group $\mathbb{Z}_2^{2n}$ are conjugate, we have that any irreducible representation $\sigma$ of the projective Pauli group can be extended to a representation $\sigma'$ of its own inertia subgroup, $I(\sigma)$, in the Clifford group. We apply the little group method to obtain the following.
\begin{theorem}\label{RepThm}
The irreducible representations of the projective Clifford group are
\begin{equation*}
        \widehat{\mathcal{C}_n} = \left\{\emph{Ind}^{\mathcal{C}_n}_{IN_n} (\sigma_1'\otimes\widetilde{\psi}) : \psi\in \widehat{IN_n/\mathbb{Z}^{2n}_2}\right\}\cup\left\{\widetilde{\psi}: \psi\in \widehat{Sp(2n,2)}\right\},
\end{equation*}
where the $\widetilde{\psi}$ in the left set is an inflation to an irreducible representation of $IN_n$ and in the right set is an inflation to an irreducible representation of $\mathcal{C}_n$.
\end{theorem}
 
\Cref{RepThm} gives a complete list of the irreducible representations of the $n$-qubit Clifford group. To actually calculate these representations, we would like to know the representations of $Sp(2n,2)$ and those of the quotient group $IN_n/\mathcal{P}_n$. Using \Cref{RepThm}, we may calculate the following example character tables.
\begin{example}
\Cref{CT:pauli1} is the character table of the 1-qubit projective Pauli group.
\begin{table}[t]
    \centering
\begin{tabular}{ |c|c c c c | }
 \hline
 {} & $[I_2]$ & $[X]$ & $[Z]$ & $[Y]$\\
 \hline
 $\psi_1$ & 1 & 1 & 1 & 1\\ 
 $\psi_2$ & 1 & -1 & 1 & -1\\ 
 $\psi_3$ & 1 & 1 & -1 & -1 \\
 $\psi_4$ & 1 & -1 & -1 & 1\\ 
 \hline
\end{tabular}
\caption{The character table of the 1-qubit projective Pauli group.}
    \label{CT:pauli1}
\end{table}
Notice that the inertia group of the representation $\psi_4$ is just the subgroup $I(\psi_4) = \langle[H],[X],[Z]\rangle\subset\mathcal{C}_1$. The extension $\psi_4'$ of $\psi_4$ to $I(\psi_4)$ is achieved by defining the value of $\psi_4'([H]) = 1$. \Cref{CT:pin1} is the character table of $I(\psi_4)/\mathbb{Z}_2^{2}$. Via GAP4 calculation, \Cref{CT:sp22} is the character table of $Sp(2,2)$. Then by \Cref{RepThm} the character table of the 1-qubit Clifford group $\mathcal{C}_1$ is \Cref{CT:cliff1}.

\begin{table}[t]
    \centering
    \begin{tabular}{ |c|c c | } 
        \hline
        {}& {$[[I_2]]$}  & {$[[H]]$} \\
        \hline
        $\phi_1$ & 1 & 1  \\ 
        $\phi_2$ & 1 & -1  \\  
        \hline
    \end{tabular}
    \caption{Character table of $I(\psi_4)/\mathbb{Z}_2^{2}$.}
    \label{CT:pin1}
\end{table}

\begin{table}[t]
    \centering
    \begin{tabular}{ |c|c c c | } 
    \hline
    {}& {$[[I_2]]$} & {$[[H]]$} & {$[[SH]]$} \\
    \hline
    $\theta_1$ & 1 & 1 & 1  \\ 
    $\theta_2$ & 1 & -1 & 1 \\ 
    $\theta_3$ & 2 & 0 & -1 \\ 
    \hline
    \end{tabular}
    \caption{Character table of $Sp(2,2)$}
    \label{CT:sp22}
\end{table}
\begin{table}[t]
\centering
\begin{tabular}{ |c|c c c c c| } 
 \hline
 {} & $[I_2]$ & {$[H]$} & {$[SH]$} & $[X]$ & {$[S]$}\\
 \hline 
 \rule{0pt}{3ex} $\widetilde{\theta}_1$ & 1 & 1 & 1 & 1 & 1\\
 $\widetilde{\theta}_2$ & 1 & -1 & 1 & 1 & -1\\ 
 $\widetilde{\theta}_3$ & 2 & 0 & -1 & 2 & 0\\
$\text{Ind}^{\mathcal{C}_1}_{I(\psi_4)}(\psi_4'\otimes\widetilde{\phi}_1)$ & 3 & -1 & 0 & -1 & 1\\
 $\text{Ind}^{\mathcal{C}_1}_{I(\psi_4)}(\psi_4'\otimes\widetilde{\phi}_2)$ & 3 & 1 & 0 & -1 & -1\\ [1ex]
 \hline
\end{tabular}
\caption{Character table of $\mathcal{C}_1$}
\label{CT:cliff1}
\end{table}
\end{example}
\subsection{The representation theory of the inertia quotient group}\label{sec:quotrep}
To understand the irreducible representations of the Clifford group with nontrivial restriction to the Pauli group, we will now examine the irreducible representations of $IN_n/\mathbb{Z}_2^{2n}$. Notice that 
\begin{equation*}
    H_1^{M} = (XZ\otimes X)H_1,
\end{equation*}
and that $H_1$ commutes with all other non-Pauli operators in the generating set of $IN_n$. From this we see that $\mathbb{Z}_2 \cong\langle[[I]],[[H_1]]\rangle$ forms an order 2 normal subgroup of $IN_n/\mathbb{Z}^{2n}_2$. For convenience we define the group $\mathcal{H}_{1,n}:=\langle [H_1],\{[X_i],[Z_i], \text{ for } i=1,\dots,n\}\rangle$. We are ready to prove the following lemma.
\begin{lemma}\label{QuotLemma}
The inertia quotient group has the affine symplectic group as a quotient group, that is
\begin{equation*}
    (IN_n/\mathbb{Z}_2^{2n})/\mathbb{Z}_2 \cong IN_n/\mathcal{H}_{1,n} \cong Sp(2(n-1),2)\ltimes\mathbb{Z}_2^{2(n-1)}.
\end{equation*}

\end{lemma}
\begin{proof}
For $x \in \mathbb{Z}_2^{2(n-1)}$, let $W_x$ be the Weyl operator defined in the proof of \Cref{QuotThm}. Consider operators of the form $X\otimes W_x$ and $Z\otimes W_x$, which we will call inertia Weyl operators since by definition $n$-qubit Weyl operators of this form are preserved by the inertia subgroup $IN_n$ under conjugation. Since the $n$-qubit Pauli group is generated by these inertia Weyl operators, the action of $U \in IN_n$ by conjugation on these operators defines the action of $U$ on the Pauli group.
 
From \Cref{QuotThm}, we know that conjugating an inertia Weyl operator by $I\otimes U$ for $U \in \text{Cliff}(n-1)$ will give us $X\otimes W_{\Gamma x}$ or $Z\otimes W_{\Gamma x}$ respectively for some $\Gamma \in Sp(2(n-1),2)$ with potential phase factors. Furthermore, we know that any such $\Gamma$ is realized by some $U \in \text{Cliff}(n-1)$. Conjugation by $H_1$ will exchange the $X$ and $Z$ on the first qubit. Conjugating by $MS_2H_2S^{-1}_2$ amounts to multiplication by $X_2$ on the left with a possible phase factor of $-1$ and a possible exchange of $X$ and $Z$ on the first qubit. Similarly conjugation by $H_1H_2MS_2H_2S^{-1}_2H_2$ amounts to multiplication by $Z_2$ on the left with a possible phase factor of $-1$ and a possible exchange of $X$ for $Z$ on the first qubit. 
 
Notice that the actions by conjugation of the matrices we have examined generate the affine symplectic group $Sp(2(n-1),2)\ltimes\mathbb{Z}_2^{2(n-1)}$ on the index $x$ of a Weyl operator $W_x$, with an extra operator $H$ that exchanges the $X$ and $Z$ on the first qubit. Since the equivalence classes of said matrices also generate $IN_n$, and the inertia Weyl operators along with $H_1$ generate $\mathcal{H}_{1,n}$, we have the result. 
\end{proof}
From the proof of this lemma, we see that the quotient group $IN_n/\mathbb{Z}_2^{2n}$ is a central extension of $Sp(2(n-1),2)\ltimes\mathbb{Z}^{2(n-1)}_2$ by $\mathbb{Z}_2$. Through GAP4 calculation we have determined that in general, the extension will not be a direct product, although it is in the two qubit case. Fix a normalized section $t$ of the central extension $IN_n/\mathbb{Z}_2^{2n}$ of $Sp(2(n-1),2)\ltimes\mathbb{Z}^{2(n-1)}_2$ by $\mathbb{Z}_2$. Let $b(h,k)\in\mathbb{Z}_2$ be the corresponding factor set. The only nontrivial irreducible representation $\xi$ of $\mathbb{Z}_2$ maps the non-identity element to $-1$. Let $\eta(h,k) = \xi(b(h,k))$. By applying \Cref{centerprop}, we obtain 
\begin{equation*}
\widehat{IN_n/\mathbb{Z}_2^{2n}} = \left\{\widetilde{\psi}:\psi\in \widehat{Sp(2(n-1),2)\ltimes\mathbb{Z}^{2(n-1)}_2}\right\}\cup\left\{\Theta:\Psi\in (\widehat{Sp(2(n-1),2)\ltimes\mathbb{Z}^{2(n-1)}_2})^{\eta}\right\},
\end{equation*}
with $\Theta$ defined by $\Theta(t(h)b) = \xi(b)\Psi(h)$ for all $h\in Sp(2(n-1),2)\ltimes\mathbb{Z}^{2(n-1)}_2$ and $b\in \mathbb{Z}_2$.
\section{Lifting irreducible characters to higher dimensional Clifford groups}\label{sec:Lifting}
We will now explain how irreducible characters of the $n$-qubit Clifford group can be used to explicitly construct characters of the $(n+1)$-qubit Clifford group. First, we need to understand the representation theory of the affine symplectic group $Sp(2n,2)\ltimes\mathbb{Z}_2^{2n}$. It is clear that if $U$ acts on $\mathbb{Z}^{2n}_2$ by $\Gamma \in Sp(2n,2)$ then $^{(\mathbf{x},\Gamma)}\sigma \sim {}^U\sigma$ for any $\sigma \in \widehat{\mathbb{Z}^{2n}_2}$ and $(\mathbf{x},\Gamma)\in Sp(2n,2)\ltimes\mathbb{Z}^{2n}_2$. Let $\sigma_1$ be the irreducible representation of $\mathbb{Z}^{2n}_2$ defined in \cref{Section:Inert}, then it follows that $I_{Sp(2n,2)\ltimes \mathbb{Z}^{2n}_2}(\sigma_1)/\mathbb{Z}^{2n}_2\cong IN_n/\mathbb{Z}^{2n}_2$. Let $\sigma_1''$ be the extension of $\sigma_1$ to $I_{Sp(2n,2)\ltimes \mathbb{Z}^{2n}_2}(\sigma_1)$ via $\sigma_1''(x,\Gamma) = \sigma_1(x)$. By applying \Cref{LittleTheorem}, we immediately obtain the following.
\begin{lemma}
The irreducible representations of the affine symplectic group are
\begin{equation*}
\widehat{Sp(2n,2)\ltimes\mathbb{Z}_2^{2n}} = \left\{\emph{Ind}^{Sp(2n,2)\ltimes\mathbb{Z}_2^{2n}}_{(IN_n/\mathbb{Z}^{2n}_2)\ltimes\mathbb{Z}_2^{2n}} (\sigma_1''\otimes\widetilde{\psi}) : \psi\in \widehat{IN_n/\mathbb{Z}^{2n}_2}\right\}\cup\left\{\widetilde{\psi}: \psi\in \widehat{Sp(2n,2)}\right\}, 
\end{equation*}
where $\widetilde{\psi}$ in the left set is the inflation to $(IN_n/\mathbb{Z}^{2n}_2)\ltimes\mathbb{Z}_2^{2n}$ and in the right set is inflation to $Sp(2n,2)\ltimes\mathbb{Z}_2^{2n}$.  
\end{lemma}
We can now prove the following lemma which was first proven by Bernd Fischer using the technique of Fischer-Clifford matrices \cite{Fischer88}.
\begin{lemma}\label{BrauerLemma}
$Sp(2n,2)\ltimes\mathbb{Z}^{2n}_2$ and $\mathcal{C}_n$ have identical character tables.
\end{lemma}
\begin{proof}
This is trivially true if $n=1$, as in that case the groups are isomorphic. For $n>1$ we first notice that 
\begin{equation*}
    (Sp(2n,2)\ltimes\mathbb{Z}^{2n}_2)/\mathbb{Z}^{2n}_2 \cong Sp(2n,2)\cong \mathcal{C}_n/\widetilde{\mathcal{P}_n}. 
\end{equation*}
The irreducible characters that come from $Sp(2n,2)$ are nothing but inflations of the irreducible characters of $Sp(2n,2)$. Thus if $\chi$ is an irreducible character of $Sp(2n,2)$ and $\widetilde{\chi}$ and $\widetilde{\chi}'$ are its inflations to $\mathcal{C}_n$ and $Sp(2n,2)\ltimes \mathbb{Z}^{2n}_2$ respectively, we have
\begin{equation}\label{eq:char}
    \widetilde{\chi}(U) = \chi(\Gamma) = \widetilde{\chi}'(\mathbf{x},\Gamma)
\end{equation}
for all $\mathbf{x} \in \mathbb{Z}^{2n}_2$ and $U\in\mathcal{C}_n$ such that $UW_{\mathbf{x}}U^{\dagger} = (-1)^{f(\mathbf{x})}W_{\Gamma\mathbf{x}}$.
 
Fix a normalized section $t: Sp(2n,2)\rightarrow \mathcal{C}_n$ of the extension 
\begin{equation*}
    1\rightarrow\mathbb{Z}_2^{2n}\rightarrow\mathcal{C}_n\rightarrow Sp(2n,2)\rightarrow 1
\end{equation*} 
such that $\sigma_1'(t(\Gamma)) = 1$ for all $\Gamma\in IN_n/\mathbb{Z}_2^{2n}$. Define the mapping $\phi: Sp(2n,2)\ltimes\mathbb{Z}^{2n}_2\rightarrow \mathcal{C}_n$ by $\phi(\mathbf{x},\Gamma) = W_{\mathbf{x}}t(\Gamma)$. It is clear that this mapping is one-to-one and onto, and $\sigma_1''(s) = \sigma_1'(\phi(s))$ for all $s \in I_{Sp(2n,2)\ltimes \mathbb{Z}^{2n}_2}(\sigma_1)$. Using the notation of \cref{eq:char} we see that $\widetilde{\chi}(\phi(s)) = \widetilde{\chi}'(s)$ for all $s\in Sp(2n,2)\ltimes \mathbb{Z}^{2n}_2$.
Let $\psi$ be an irreducible representation of $IN_n/\mathbb{Z}^{2n}_2$, and $\widetilde{\psi}$ and $\widetilde{\psi}'$ be its inflations to $I_{\mathcal{C}_n}(\sigma_1)$ and $I_{Sp(2n,2)\ltimes \mathbb{Z}^{2n}_2}(\sigma_1)$ respectively. From the formula for induced characters, we have
\begin{equation*}
    \text{Ind}^{Sp(2n,2)\ltimes\mathbb{Z}^{2n}_2}_{I_{Sp(2n,2)\ltimes\mathbb{Z}^{2n}_2}(\sigma_1)}(\widetilde{\psi}'\otimes\sigma_1'')(s) = \frac{1}{|I_{Sp(2n,2)\ltimes\mathbb{
    Z}^{2n}_2}(\sigma_1)|}\sum_{{\substack{r\in Sp(2n,2)\ltimes\mathbb{Z}^{2n}_2\\r^{-1}sr \in I_{Sp(2n,2)\ltimes\mathbb{Z}^{2n}_2}(\sigma_1)}}}\widetilde{\psi}'\otimes\sigma_1''(r^{-1}sr),
\end{equation*}
and
\begin{equation*}
    \text{Ind}^{\mathcal{C}_n}_{I_{\mathcal{C}_n}(\sigma_1)}(\widetilde{\psi}\otimes\sigma_1')(s) = \frac{1}{|I_{\mathcal{C}_n}(\sigma_1)|}\sum_{\substack{r\in \mathcal{C}_n\\r^{-1}sr \in I_{\mathcal{C}_n}(\sigma_1)}}\widetilde{\psi}\otimes\sigma_1'(r^{-1}sr).
\end{equation*}
Since the action by conjugation of $\phi(\mathbf{x},\Gamma)$ depends only on $\Gamma$, we see that $\phi(r)^{-1}\phi(s)\phi(r)\in I_{\mathcal{C}_n(\sigma_1)}$ if and only if $r^{-1}sr\in I_{Sp(2n,2)\ltimes\mathbb{Z}^{2n}_2}(\sigma_1)$ for any $r,s\in Sp(2n,2)\ltimes\mathbb{Z}^{2n}_2$, and furthermore $\widetilde{\psi}(\phi(r)^{-1}\phi(s)\phi(r)) = \widetilde{\psi}'(r^{-1}sr)$. Finally, we obtain 
\begin{equation*}
    \text{Ind}^{\mathcal{C}_n}_{I_{\mathcal{C}_n}(\sigma_1)}(\widetilde{\psi}\otimes\sigma_1')(\phi(s)) = \text{Ind}^{Sp(2n,2)\ltimes\mathbb{Z}^{2n}_2}_{I_{Sp(2n,2)\ltimes\mathbb{Z}^{2n}_2}(\sigma_1)}(\widetilde{\psi}'\otimes\sigma_1'')(s)
\end{equation*}
for all $s \in Sp(2n,2)\ltimes\mathbb{Z}^{2n}_2$.
By column orthogonality of character tables we have that $r,s\in Sp(2n,2)\ltimes\mathbb{Z}^{2n}_2$ are conjugate if and only if $\phi(r)$ and $\phi(t)$ are conjugate in $\mathcal{C}_n$. Thus the map $\phi$ respects conjugacy classes and the character tables are identical.
\end{proof}
Taken together these lemmas imply a remarkable property of the Clifford group.
\begin{theorem}
Let $\phi:Sp(2{n},2)\ltimes\mathbb{Z}^{2{n}}_2\rightarrow\mathcal{C}_{n}$ be the map defined in the proof of \Cref{BrauerLemma}. If $\chi$ is an irreducible character of the $n$-qubit Clifford group $\mathcal{C}_n$ then \emph{Ind}${}^{\mathcal{C}_{n+1}}_{IN_{n+1}}\widetilde{(\chi\circ\phi)}\otimes\sigma_1'$ is an irreducible character of the $(n+1)$-qubit Clifford group $\mathcal{C}_{n+1}$.
\end{theorem}
\begin{proof}
By \Cref{BrauerLemma} we see that every irreducible character $\chi$ of $\mathcal{C}_n$ is also an irreducible character of $Sp(2n,2)\ltimes\mathbb{Z}^{2n}_2$ when precomposed with the bijection $\phi$ of the conjugacy classes of the two groups. We can then see by \Cref{QuotLemma} that the irreducible character $\chi_{\phi} := \chi\circ \phi$ of $Sp(2n,2)\ltimes\mathbb{Z}^{2n}_2$ inflates to an irreducible character $\widetilde{\chi_{\phi}}$ of $IN_{n+1}$ that contains $\mathcal{H}_{1,n+1}$ in its kernel. In Particular this means that $\widetilde{\mathcal{P}}_{n+1}$ will be contained in the kernel of $\widetilde{\chi_{\phi}}$, so we know that $\widetilde{\chi_{\phi}}\otimes\sigma_1'$ is an irreducible character of $IN_{n+1}$ that has $\sigma_1$ in the decomposition of its restriction to $\widetilde{\mathcal{P}}_{n+1}$ into irreducible representations. Therefore, by the Clifford correspondence we obtain the result.
\end{proof}
This gives a straightforward method for obtaining irreducible characters of the $(n+1)$-qubit Clifford group from irreducible characters of the $n$-qubit Clifford group.
\begin{table}[b]
    \centering
    \begin{tabular}{c|c}
        1 & $[I_4]$ \\
        2 & $[I_2\otimes Z]$ \\
        3 & $[I_2\otimes H]$ \\
        4 & $[I_2\otimes HZ]$ \\
        5 & $[Z\otimes HZ]$ \\
        6 & $[H\otimes H]$ \\
        7 & $[H\otimes HZ]$ \\
        8 & $[HZ\otimes HZ]$ \\
        9 & $[(S^{-1}H\otimes SH) CZ (S^{-1}HSH\otimes ZH)]$ \\
        10 & $[(S^{-1}H\otimes SHS^{-1})CZ(ZHS\otimes HSH)]$ \\
        11 & $[I_2\otimes SHS^{-1}XS^{-1}]$ \\
        12 & $[HS^{-1}XS^{-1}H\otimes S^{-1}H]$ \\
        13 & $[H\otimes SHS^{-1}XS^{-1}H]$ \\
        14 & $[S^{-1}XS^{-1}H\otimes S^{-1}H]$ \\
        15 & $[(SH\otimes SHS^{-1})CZ (ZHSH\otimes I_2)]$ \\
        16 & $[(SH\otimes ZHSH)CZ(SHS\otimes I_2)]$ \\
        17 & $[(I_2\otimes SH)CZ(H\otimes I_2)CZ(SXSH\otimes I_2)]$ \\
        18 & $[(I_2\otimes S^{-1}HSX)CZ(H\otimes I_2)CZ(H\otimes S)]$ \\
        19 & $[CZ(I_2\otimes S)(H\otimes H)CZ(H \otimes H)CZ(ZX\otimes H)]$ \\
        20 & $[(S^{-1}\otimes I_2)CZ(H\otimes HSH)CZ(SXSH\otimes I_2)]$ \\
        21 & $[(I_2\otimes S^{-1}XS^{-1})CZ(SH\otimes H)CZ (I_2\otimes S)]$ 
    \end{tabular}
\caption{Conjugacy class representatives for $
\mathcal{C}_2$}
    \label{classreps}
\end{table}
\begin{example}
    As an example, we demonstrate the lifting procedure from the 1-qubit to the 2-qubit Clifford group. In this case, because of the isomorphism $\mathcal{C}_1 \cong Sp(2,2)\ltimes\mathbb{Z}^2_2$, we know that the inertia quotient $IN_2/\mathbb{Z}^4_2$ group is a central extension of $\mathcal{C}_1$ by $\mathbb{Z}_2$. Moreover, in this case, the extension splits and we have $IN_2/\mathbb{Z}^2_2 \cong \mathcal{C}_1\times\mathbb{Z}_2$. \Cref{CT:pin2} is the character table of $\mathcal{C}_1\times\mathbb{Z}_2$, where we denote the characters of $\mathcal{C}_1$ by $\chi_i$ for $i \in\{1,\dots 5\}$ and denote the characters of $\mathbb{Z}_2$ by $\theta_1$ and $\theta_2$, with $\theta_1$ being the trivial representation.
\begin{table}[t]
    \centering
\begin{tabular}{ |c|c c c c c c c c c c| } 
 \hline
 {} & $([I_2],0)$ & {$([I_2],1)$} & {$([H],1)$}& {$([H],0)$} & {$([SH],1)$} & {$([SH],0)$} & {$([X],1)$} & {$([X],0)$} & {$([S],1)$} & {$([S],0)$}\\
 \hline
 $\mu_1:=\chi_1\times\theta_1$ & 1 & 1 & 1 & 1 & 1 & 1 & 1 & 1 & 1 & 1\\ 
 $\mu_2:=\chi_1\times\theta_2$ & 1 & -1 & -1 & 1 & -1 & 1 & -1 & 1 & -1 & 1\\
 $\mu_3:=\chi_2\times\theta_2$ & 1 & -1 & 1 & -1 & -1 & 1 & -1 & 1 & 1 & -1\\
 $\mu_4:=\chi_2\times\theta_1$ & 1 & 1 & -1 & -1 & 1 & 1 & 1 & 1 & -1 & -1\\
 $\mu_5:=\chi_3\times\theta_2$ & 2 & -2 & 0 & 0 & 1 & -1 & -2 & 2 & 0 & 0\\
 $\mu_6:=\chi_3\times\theta_1$ & 2 & 2 & 0 & 0 & -1 & -1 & 2 & 2 & 0 & 0\\
 $\mu_7:=\chi_4\times\theta_2$ & 3 & -3 & -1 & 1 & 0 & 0 & 1 & -1 & 1 & -1\\
 $\mu_8:=\chi_5\times\theta_2$ & 3 & -3 & 1 & -1 & 0 & 0 & 1 & -1 & -1 & 1\\
 $\mu_9:=\chi_4\times\theta_1$ & 3 & 3 & -1 & -1 & 0 & 0 & -1 & -1 & 1 & 1\\ 
 $\mu_{10}:=\chi_5\times\theta_1$ & 3 & 3 & 1 & 1 & 0 & 0 & -1 & -1 & -1 & -1\\ 
 \hline
\end{tabular}
\caption{Character table of $\mathcal{C}_1\times \mathbb{Z}_2$}
\label{CT:pin2}
\end{table}
    So from every character of $\mathcal{C}_1$ we will get two characters of $\mathcal{C}_2$, and the character table of $\mathcal{C}_2$ is determined entirely by these characters, and inflated characters from $Sp(4,2)$. Thus the character table of $\mathcal{C}_2$ is \Cref{CT:cliff2}, where the $\widetilde{\psi_i}$ are inflated characters from $Sp(4,2)$. 
    \begin{table}[t]
        \centering
    \tiny\begin{tabular}{|c|c c c c c c c c c c c c c c c c c c c c c|}
    \hline
     & $1$ & $2$ & $3$ & $4$ & $5$ & $6$ & $7$ & $8$ & $9$ & ${10}$ & ${11}$ & ${12}$ & ${13}$ & ${14}$ & ${15}$ & ${16}$ & ${17}$ & ${18}$ & ${19}$ & ${20}$ & ${21}$ \\
     \hline
     & & & & & & & & & & & & & & & & & & & & & 
     \\
        $\widetilde{\psi_1}$ & 1 & 1 & 1 & 1 & 1 & 1 & 1 & 1 & 1 & 1 & 1 & 1 & 1 & 1 & 1 & 1 & 1 & 1 & 1 & 1 & 1 \\
        $\widetilde{\psi_2}$ & 1 & 1 & -1 & -1 & -1 & 1 & 1 & 1 & -1 & -1 & 1 & 1 & -1 & -1 & -1 & -1 & 1 & 1 & -1 & 1 & 1\\ 
        $\widetilde{\psi_3}$ & 5 & 5 & 3 & 3 & 3 & 1 & 1 & 1 & -1 & -1 & 2 & 2 & 0  & 0 & 1 & 1 & -1 & -1 & -1 & 0 &-1\\
        $\widetilde{\psi_4}$ & 5 & 5 & -3 & -3 & -3 & 1 & 1 & 1 & 1 & 1 & 2 & 2 & 0 & 0 & -1 & -1 & -1 & -1 & 1 & 0 & -1\\
        $\widetilde{\psi_5}$ & 5 & 5 & -1 & -1 & -1 & 1 & 1 & 1 & 3 & 3 & -1 & -1 & -1 & -1 & 1 & 1 & -1 & -1 & 0 & 0 & 2\\
        $\widetilde{\psi_6}$ & 5 & 5 & 1 & 1 & 1 & 1 & 1 & 1 & -3 & -3 & -1 & -1 & 1  & 1 & -1 & -1 & -1 & -1 & 0 & 0 & 2\\
        $\widetilde{\psi_7}$ & 9 & 9 & -3 & -3 & -3 & 1 & 1 & 1 & -3 & -3 & 0 & 0 & 0 & 0 & 1 & 1 & 1 & 1 & 0 & -1 & 0\\
        $\widetilde{\psi_8}$ & 9 & 9 & 3 & 3 & 3 & 1 & 1 & 1 & 3 & 3 & 0 & 0 & 0 & 0 & -1 & -1 & 1 & 1 & 0 & -1 & 0\\
        $\widetilde{\psi_9}$ & 10 & 10 & -2 & -2 & -2 & -2 & -2 & -2 & 2 & 2 & 1 & 1 & 1 & 1 & 0 & 0 & 0 & 0 & -1 & 0 & 1\\
        $\widetilde{\psi_{10}}$ & 10 & 10 & 2 & 2 & 2 & -2 & -2 & -2 & -2 & -2 & 1 & 1 & -1 & -1 & 0 & 0 & 0 & 0 & 1 & 0 & 1\\
        $\emph{Ind}^{\mathcal{C}_2}_{I_{\mathcal{C}_2}(\sigma_1)}(\widetilde{\mu_1}\otimes\sigma_1')$ & 15 & -1 & 1 & 5 & -3 & 3 & -1 & -1 & 1 & -3 & 3 & -1 & 1 & -1 & 1 & -1 & 1 & -1 & 0 & 0 & 0 \\
        $\emph{Ind}^{\mathcal{C}_2}_{I_{\mathcal{C}_2}(\sigma_1)}(\widetilde{\mu_2}\otimes\sigma_1')$ & 15 & -1 & 1 & -7 & 1 & -1 & -1 & 3 & 1 & -3 & 3 & -1 & 1 & -1 & -1 & 1 & -1 & 1 & 0 & 0 & 0  \\
        $\emph{Ind}^{\mathcal{C}_2}_{I_{\mathcal{C}_2}(\sigma_1)}(\widetilde{\mu_3}\otimes\sigma_1')$ & 15 & -1 & -1 & 7 & -1 & -1 & -1 & 3 & -1 & 3 & 3 & -1 & -1 & 1 & 1 & -1 & -1 & 1 & 0 & 0 & 0\\
        $\emph{Ind}^{\mathcal{C}_2}_{I_{\mathcal{C}_2}(\sigma_1)}(\widetilde{\mu_4}\otimes\sigma_1')$ & 15 & -1 & -1 & -5 & 3 & 3 & -1 & -1 & -1 & 3 & 3 & -1 & -1 & 1 & -1 & 1 & 1 & -1 & 0 & 0 & 0\\
        $\widetilde{\psi_{11}}$ & 16 & 16 & 0 & 0 & 0 & 0 & 0 & 0 & 0 & 0 & -2 & -2 & 0 & 0 & 0 & 0 & 0 & 0 & 0 & 1 & -2
        \\
        $\emph{Ind}^{\mathcal{C}_2}_{I_{\mathcal{C}_2}(\sigma_1)}(\widetilde{\mu_5}\otimes\sigma_1')$ & 30 & -2 & -2 & 2 & 2 & 2 & -2 & 2 & -2 & 6 & -3 & 1 & 1 & -1 & 0 & 0 & 0 & 0 & 0 & 0 & 0\\
        $\emph{Ind}^{\mathcal{C}_2}_{I_{\mathcal{C}_2}(\sigma_1)}(\widetilde{\mu_6}\otimes\sigma_1')$ & 30 & -2 & 2 & -2 & -2 & 2 & -2 & 2 & 2 & -6 & -3 & 1 & -1 & 1 & 0 & 0 & 0 & 0 & 0 & 0 & 0
        \\
        $\emph{Ind}^{\mathcal{C}_2}_{I_{\mathcal{C}_2}(\sigma_1)}(\widetilde{\mu_7}\otimes\sigma_1')$ & 45 & -3 & 3 & -9 & -1 & -3 & 1 & 1 & -1 & 3 & 0 & 0 & 0 & 0 & 1 & -1 & 1 & -1 & 0 & 0 & 0
        \\
        $\emph{Ind}^{\mathcal{C}_2}_{I_{\mathcal{C}_2}(\sigma_1)}(\widetilde{\mu_8}\otimes\sigma_1')$ & 45 & -3 & -3 & 9 & 1 & -3 & 1 & 1 & 1 & -3 & 0 & 0 & 0 & 0 & -1 & 1 & 1 & -1 & 0 & 0 & 0
        \\
        $\emph{Ind}^{\mathcal{C}_2}_{I_{\mathcal{C}_2}(\sigma_1)}(\widetilde{\mu_9}\otimes\sigma_1')$ & 45 & -3 & -3 & -3 & 5 & 1 & 1 & -3 & 1 & -3 & 0 & 0 & 0 & 0 & 1 & -1 & -1 & 1 & 0 & 0 & 0
        \\
        $\emph{Ind}^{\mathcal{C}_2}_{I_{\mathcal{C}_2}(\sigma_1)}(\widetilde{\mu_{10}}\otimes\sigma_1')$ & 45 & -3 & 3 & 3 & -5 & 1 & 1 & -3 & -1 & 3 & 0 & 0 & 0 & 0 & -1 & 1 & -1 & 1 & 0 & 0 & 0\\
        & & & & & & & & & & & & & & & & & & & & &\\
        \hline
    \end{tabular}
\caption{Character table of $\mathcal{C}_2$}
        \label{CT:cliff2}
    \end{table}
The numbered conjugacy classes in the character table of $\mathcal{C}_2$ are represented by the elements in \Cref{classreps}.

\end{example}

% If in two-column mode, this environment will change to single-column format so that long equations can be displayed. 
% Use only when necessary.
%\begin{widetext}
%$$\mbox{put long equation here}$$
%\end{widetext}

% Figures should be put into the text as floats. 
% Use the graphics or graphicx packages (distributed with LaTeX2e).
% See the LaTeX Graphics Companion by Michel Goosens, Sebastian Rahtz, and Frank Mittelbach for examples. 
%
% Here is an example of the general form of a figure:
% Fill in the caption in the braces of the \caption{} command. 
% Put the label that you will use with \ref{} command in the braces of the \label{} command.
%
% \begin{figure}
% \includegraphics{}%
% \caption{\label{}}%
% \end{figure}

% Tables may be be put in the text as floats.
% Here is an example of the general form of a table:
% Fill in the caption in the braces of the \caption{} command. Put the label
% that you will use with \ref{} command in the braces of the \label{} command.
% Insert the column specifiers (l, r, c, d, etc.) in the empty braces of the
% \begin{tabular}{} command.
%
% \begin{table}
% \caption{\label{} }
% \begin{tabular}{}
% \end{tabular}
% \end{table}

% If you have acknowledgments, this puts in the proper section head.
\begin{acknowledgments}
I thank William Slofstra, Benjamin Lovitz and Jack Davis for helpful discussions. I acknowledge the support of the Natural Sciences and Engineering Research Council of Canada (NSERC). In particular, this work was supported by an NSERC CGS D award.
\end{acknowledgments}
%\clearpage
% Create the reference section using BibTeX:
\bibliography{citations}

\end{document}